\pdfoutput=1
\documentclass[11pt,a4paper]{article}
\usepackage[utf8]{inputenc}

\usepackage{lmodern}
\usepackage[T1]{fontenc}

\usepackage{hyperref}
\hypersetup{breaklinks, urlcolor=blue, colorlinks=true, citecolor=green!50!black, linkcolor=blue}
\usepackage[letterpaper, left=1in, right=1in, top=0.9in, bottom=0.9in]{geometry}
\usepackage[utf8]{inputenc}
\usepackage[american]{babel}
\usepackage{etoolbox}

\usepackage[normalem]{ulem}

\usepackage{graphicx}
\graphicspath{{images/}{../images/}}

\usepackage{amsthm, amsmath, amssymb, cases}
\usepackage[capitalize]{cleveref}
\usepackage{thmtools}
\usepackage{theoremstyles}
\usepackage[shortlabels]{enumitem}
\usepackage{mdframed}
\usepackage{bbm}
\usepackage{bm}

\usepackage[ruled]{algorithm}
\usepackage[noend]{algpseudocode}

\usepackage{microtype}

\usepackage{mdframed}
\usepackage{xcolor}

\usepackage{tikz}
\usetikzlibrary{fit, arrows, shapes, matrix, chains, positioning, backgrounds, arrows.meta, tikzmark, decorations.pathreplacing}

\usepackage{modletters}

\usepackage[sort,nocompress]{cite}
\usepackage{authblk}

\usepackage{hyperref}
\hypersetup{breaklinks, urlcolor=blue, colorlinks=true, citecolor=green!50!black, linkcolor=blue}
\usepackage[letterpaper, left=1in, right=1in, top=0.9in, bottom=0.9in]{geometry}
\usepackage[utf8]{inputenc}
\usepackage[american]{babel}
\usepackage{etoolbox}

\usepackage[normalem]{ulem}

\usepackage{graphicx}
\graphicspath{{images/}{../images/}}

\usepackage{amsthm, amsmath, amssymb, cases}
\usepackage[capitalize]{cleveref}
\usepackage{thmtools}
\usepackage{theoremstyles}
\usepackage[shortlabels]{enumitem}
\usepackage{mdframed}
\usepackage{bbm}
\usepackage{bm}

\usepackage[ruled]{algorithm}
\usepackage[noend]{algpseudocode}

\usepackage{microtype}

\usepackage{mdframed}
\usepackage{xcolor}

\usepackage{tikz}
\usetikzlibrary{fit, arrows, shapes, matrix, chains, positioning, backgrounds, arrows.meta, tikzmark, decorations.pathreplacing}

\usepackage{modletters}

\newcommand\eps\varepsilon
\DeclareMathOperator{\rk}{rk}

\title{Breaking $O(nr)$ for Matroid Intersection}

\author{Joakim Blikstad}
\affil{KTH Royal Institute of Technology, Sweden}
\affil{\texttt{blikstad@kth.se}}
\date{}

\date{\today}

\begin{document}

\maketitle

\begin{abstract}
We present algorithms that break the $\tilde O(nr)$-independence-query bound for the Matroid Intersection problem \emph{for the full range of $r$}; where $n$ is the size of the ground set and $r\leq n$ is the size of the largest common independent set. The $\tilde O(nr)$ bound was due to the efficient implementations [CLSSW FOCS'19; Nguy\~{\^e}n 2019] of the classic algorithm of Cunningham [SICOMP'86]. It was recently broken for large $r$ ($r=\omega(\sqrt{n})$), first by the $\tO(n^{1.5}/\eps^{1.5})$-query $(1-\eps)$-approximation algorithm of CLSSW [FOCS'19], and subsequently by the $\tO(n^{6/5}r^{3/5})$-query exact algorithm of BvdBMN [STOC'21]. No algorithm---even an approximation one---was known to break the $\tilde O(nr)$ bound for the full range of $r$.
We present an $\tO(n\sqrt{r}/\eps)$-query $(1-\eps)$-approximation algorithm and an $\tO(nr^{3/4})$-query exact algorithm. Our algorithms improve the $\tilde O(nr)$ bound and also the bounds by CLSSW and BvdBMN for the full range of $r$.

\end{abstract}

\section{Introduction}
\textsf{\textbf{Matroid Intersection}} is a fundamental problem in combinatorial optimization that has been studied for more than half a century.
The classic version of this problem is as follows:
\emph{Given two matroids $\cM_1 = (V, \cI_1)$ and $\cM_2 = (V, \cI_2)$ over a common ground set $V$ of $n$ elements, find the largest common independent set $S^*\in \cI_1 \cap \cI_2$ by making  independence oracle queries\footnote{There are also other oracle models considered in the literature (e.g.\ rank-oracles), but in this paper we focus on the independence query model.
Whenever we say \emph{query} in this paper, we thus mean \emph{independence query}.} of the form ``Is $S \in \cI_1$?'' or ``Is $S \in \cI_2$?'' for $S \subseteq V$.}
The size of the largest common independent set is 
usually denoted by $r$.

Matroid intersection can be used to model many other combinatorial optimization problems,
such as bipartite matching, arborescences, spanning tree packing, etc.
As such, designing algorithms for matroid intersection is an interesting problem to study.

In this paper, we consider the task of finding a $(1-\eps)$-approximate solution to the matroid intersection problem, that is finding some common independent set $S$ of size at least $(1-\eps)r$.
We show an improvement of approximation algorithms for matroid intersection, and
as a consequence also obtain an improvement for the \emph{exact} matroid intersection problem.

\paragraph{Previous work.}
Polynomial algorithms for matroid intersection started with the work of Edmond's
$O(n^2r)$-query algorithms
\cite{edmonds1968matroid,edmonds1970submodular,edmonds1979matroid}
in the 1960s.
Since then, there has been a long line of research
e.g.\ \cite{aignerD, Lawler75, cunningham1986improved, lee2015faster, ChekuriQ16, chakrabarty2019faster,  quadratic2021}.
Cunningham \cite{cunningham1986improved} designed a
 $O(nr^{1.5})$-query blocking-flow algorithm
in 1986, similar to that of Hopcroft-Karp's bipartite-matching or Dinic's maximum-flow algorithms.
Chekuri and Quanrud \cite{ChekuriQ16} pointed out that
Cunningham's classic algorithm \cite{cunningham1986improved} from 1986 
is already a $O(nr/\eps)$-query $(1-\epsilon)$-approximation algorithm.
Recently, 
Chakrabarty-Lee-Sidford-Singla-Wong
\cite{chakrabarty2019faster} and Nguy\~{\^e}n \cite{nguyen2019note}
independently showed how to implement Cunningham's classic algorithm using only $\tO(nr)$
independence queries. This is akin to spending $\tilde O(n)$ queries to find each of the so-called \emph{augmenting paths}. 
A fundamental question is whether several augmenting paths can be found simultaneously to break the $\tilde O(nr)$ bound.

This question has been answered for large $r$ ($r=\omega(\sqrt{n})$), first by the $\tO(n^{1.5}/\eps^{1.5})$-query $(1-\epsilon)$-approximation algorithm of
Chakrabarty-Lee-Sidford-Singla-Wong\footnote{In the same paper they also show a $\tO(n^{2}r^{-1}\eps^{-2} + r^{1.5}\eps^{-4.5})$-query algorithm.} \cite{chakrabarty2019faster}, and very recently by the randomized $\tO(n^{6/5}r^{3/5})$-query exact algorithm of Blikstad-v.d.Brand-Mukhopadhyay-Nanongkai \cite{quadratic2021}. 
Whether we can break the $O(nr)$-query bound {\em for the full range of $r$} remained open \emph{even for approximation algorithms}.

\paragraph{Our results.}
We break the $O(nr)$-query bound for both \emph{approximation}
and \emph{exact} algorithms.
We first state our results for approximate matroid intersection.\footnote{
The $\tO(n^{2}r^{-1}\eps^{-2} + r^{1.5}\eps^{-4.5})$-query algorithm of \cite{chakrabarty2019faster} is the only previous algorithm which is more efficient than our algorithm in some range of $r$ and $\eps$. Actually, since the $\tO(n^{2}r^{-1}\eps^{-2} + r^{1.5}\eps^{-4.5})$-query algorithm use the $\tO(n^{1.5}/\eps^{1.5})$ algorithm as a subroutine, we do get
a slightly improved version by using our $\tO(n\sqrt{r}/\eps)$ algorithm as the subroutine instead:
$\tO(n^{2}r^{-1}\eps^{-2} + r^{1.5}\eps^{-4})$.}

\begin{restatable}[Approximation algorithm]{theorem}{ThmApprox} \label{thm:approx}
There is a deterministic algorithm
which given two matroids $\cM_1 = (V,\cI_1)$ and $\cM_2 = (V, \cI_2)$ on the same ground set $V$, finds a common independent set $S\in \cI_1\cap \cI_2$ with $|S| \ge (1-\eps) r$,
using $O\left(\frac{n\sqrt{r\log r}}{\eps}\right)$ independence queries.
\end{restatable}

Plugging Theorem \ref{thm:approx} in the framework of \cite{quadratic2021}, we get an improved algorithm---more efficient than the previous state-of-the-art---for exact matroid intersection which we state next.

\begin{restatable}[Exact algorithm]{theorem}{ThmExact} \label{thm:exact}
There is a randomized algorithm
which given two matroids $\cM_1 = (V,\cI_1)$ and $\cM_2 = (V, \cI_2)$ on the same ground set $V$, finds a common independent set $S\in \cI_1\cap \cI_2$ of maximum cardinality $r$, 
and w.h.p.\footnote{\emph{w.h.p. = with high probability} meaning with probability $1-n^{-c}$ for some arbitrarily large constant $c$.} 
uses $O(nr^{3/4}\log n)$ independence queries.
There is also a deterministic exact algorithm using $O(nr^{5/6}\log n)$ queries.
\end{restatable}

\begin{remark}
Although we only focus on the query-complexity in this paper, we note that the
time-complexity of the algorithms are dominated by query-oracle calls.
That is, our approximation algorithm runs in
$\tO(n\sqrt{r}\mathcal{T}_{\text{ind}}/\eps)$  time, 
and the exact algorithms in $\tO(nr^{3/4} \mathcal{T}_{\text{ind}})$ (randomized)
respectively $\tO(nr^{5/6} \mathcal{T}_{\text{ind}})$ time (deterministic),
where $\mathcal{T}_{\text{ind}}$ denotes the time-complexity of the independence-oracle.
\end{remark}

\subsection{Technical Overview} \label{sec:overview}
\paragraph{Approximation algorithm.}
Our approximation algorithm (\cref{thm:approx}) is a modified version of
Chakrabarty-Lee-Sidford-Singla-Wong's $\tO(n^{1.5}/\eps^{1.5})$-query approximation algorithm
\cite[Section~6]{chakrabarty2019faster}.
The algorithm is based on the ideas of Cunningham's classic blocking-flow algorithm \cite{cunningham1986improved} and runs in $O(1/\eps)$ phases, where in each phase the algorithm seeks to find a \emph{maximal} set of \emph{augmentations} in the \emph{exchange graph}.
Given a common independent set $S\in \cI_1\cap \cI_2$, the \emph{exchange graph} $G(S)$ is a directed bipartite graph (with bipartition $(S+\{s,t\}, V\setminus S)$).
Finding a shortest $(s,t)$-path, called an \emph{augmenting path}, in $G(S)$
means one can increase the size of the common independent set $S$ by 1.
Since the exchange graph changes after each augmentation,\footnote{Unlike what happens in augmenting path algorithms for flow and bipartite matching, where the underlying graphs remain the same.}
 and we do not know how to find a single augmenting path faster than $\Omega(n)$ queries, 
the need to find several augmentations in parallel arises.
\cite[Section~6]{chakrabarty2019faster} introduces the notion of \emph{augmenting sets}: a generalization of the classical \emph{augmenting paths} but where one can perform many augmentations in parallel.

So the revised goal of the algorithm is to, in each phase, efficiently find a \emph{maximal augmenting set} (akin to a \emph{blocking-flow} in bipartite matching or flow algorithms).
Towards this goal, the algorithm
maintains a relaxed version of augmenting set---called a \emph{partial} augmenting set---and keeps \emph{refining} it to make it ``better'' (i.e.\ closer to a maximal augmenting set).
Here we give two independent improvements on top of the algorithm of \cite{chakrabarty2019faster}:
\begin{enumerate}
    \item \label{itm:approx-improvement-r}
     The algorithm of \cite{chakrabarty2019faster} refines the partial augmenting set by a sequence of operations on two adjacent distance layers in the exchange graph.
    In our algorithm, we instead consider \emph{three} consecutive layers for our basic refinement procedures.
    This lets us focus our analysis on what happens in $S$---the ``left'' side of the bipartite exchange graph---which contains at most $r$ elements in total (in contrast to \cite{chakrabarty2019faster} where the performance analysis is dependent on all $n$ elements). The number of times we need to run the refinement procedures thus depends on $r$, instead of $n$, which makes the algorithm faster when $r = o(n)$.
    \item
    \label{itm:approx-improvement-eps}
    When the partial augmenting set is ``close enough'' to a maximal augmenting set,
    \cite{chakrabarty2019faster} falls back to finding the remaining augmenting paths one at a time.
    In our algorithm, we also change to a different procedure when the partial augmenting
    set is close enough to maximal.
    The difference is that, instead of finding arbitrary augmenting paths, we find a special type of \emph{valid paths} with respect to the partial augmenting set, so that these paths can be used to further improve (refine) the partial augmenting set.
    The number of valid paths we need to find is less than the number of augmenting paths
    \cite{chakrabarty2019faster} needs to find.
    This decreases the dependency on $\eps$ in the final algorithm.
\end{enumerate}
The first improvement (\cref{itm:approx-improvement-r}) replaces the $\sqrt{n}$ term
with a $\sqrt{r}$ term in the query complexity of the algorithm.
The second improvement (\cref{itm:approx-improvement-eps}) shaves off a $1/\sqrt{\eps}$ term from the query complexity.
Together they thus bring down the query complexity from $\tO(\frac{n\sqrt{n}}{\eps \sqrt{\eps}})$ in \cite{chakrabarty2019faster} to 
$\tO(\frac{n\sqrt{r}}{\eps})$ as in our \cref{thm:approx}.
Note that these two improvements are independent of each other, and can be applied individually.

\paragraph{Exact algorithm.}
To obtain the exact algorithm (\cref{thm:exact}), we use
the framework of
Blikstad-v.d.Brand-Mukhopadhyay-Nanongkai's $\tO(n^{6/5}r^{3/5})$-query exact algorithm
\cite{quadratic2021}.
The main idea of this algorithm is to combine approximation algorithms---which can efficiently find a common independent set only $\eps r$ away from the optimal---with a randomized $\tO(n\sqrt{r})$-query subroutine to find each of the remaining \emph{few, very long} augmenting paths.
The $\tO(n^{6/5}r^{3/5})$-query exact algorithm
\cite{quadratic2021} currently uses Chakrabarty-Lee-Sidford-Singla-Wong's $\tO(n^{1.5}/\eps^{1.5})$ approximation algorithm \cite{chakrabarty2019faster} as
a subroutine. Simply replacing it with our improved approximation algorithm
(\cref{thm:approx}) yields our $\tO(nr^{3/4})$-query exact algorithm.

\section{Preliminaries} \label{sec:prelim-real}
We use the standard definitions of \emph{matroid} $\cM = (V,\cI)$;
\emph{rank} $\rk(X)$ for any $X\subseteq V$; \emph{exchange graph} $G(S)$ for a common independent set $S\in \cI_1\cap \cI_2$; and \emph{augmenting paths} in $G(S)$ throughout this paper.  
For completeness, we define them below.
We also need the notions of \emph{augmenting sets} introduced by \cite{chakrabarty2019faster}, which we also define in later this section.

\subsection*{Matroids}

\begin{definition}[Matroid]
A \emph{matroid} is a tuple $\cM = (V,\cI)$ of a \emph{ground set} $V$ of $n$ elements, and  non-empty family $\cI\subseteq 2^V$ of \emph{independent sets}
satisfying
\begin{description}
\item[Downward closure:] if $S\in\cI$, then $S'\in \cI$ for all $S'\subseteq S$.
\item[Exchange property:] if $S, S'\in \cI$, $|S| > |S'|$, then there exists $x\in S\setminus S'$ such that $S'\cup\{x\}\in\cI$.
\end{description}
\end{definition}

\begin{definition}[Set notation]
We will use $A+x$ and $A-x$ to denote $A\cup \{x\}$ respectively
$A\setminus \{x\}$, as is usual in matroid intersection literature.
We will also use $\bar{A} := V\setminus A$, $A+B := A\cup B$, and $A-B := A\setminus B$.
\end{definition}

\begin{definition}[Matroid rank]
The \emph{rank} of $A\subseteq V$, denoted by $\rk(A)$, is the size of the largest (or, equivalently, any maximal) independent set contained in $A$. It is well-known that
the rank-function is submodular,
i.e.\ $\rk(A+x) - \rk(A) \ge \rk(B+x) - \rk(B)$
whenever $A\subseteq B\subseteq V$ and $x\in V\setminus B$.\footnote{Usually denoted as the \emph{diminishing returns} property of submodular functions.}
Note that $\rk(A) = |A|$ if and only if $A\subseteq \cI$.
\end{definition}

\begin{definition}[Matroid Intersection]
Given two matroids $\cM_1 = (V, \cI_1)$
and $\cM_2 = (V, \cI_2)$ over the same ground set $V$,
a \emph{common independent set} $S$ is a set in $\cI_1\cap \cI_2$.
The \emph{matroid intersection problem} asks us to find
the largest common independent set---whose cardinality we denote by $r$.
We use $\rk_1$ and $\rk_2$ to be the rank functions of the corresponding matroids.
\end{definition}

\subsection*{The Exchange Graph} \label{sec:prelim-graph}
Many matroid intersection algorithms, e.g. those in \cite{edmonds1979matroid, aignerD, Lawler75, cunningham1986improved, nguyen2019note, quadratic2021}, are based on iteratively
finding \emph{augmenting paths} in the \emph{exchange graph}.

\begin{definition}[Exchange graph] \label{def:exchange}
Given two matroids $\cM_1 = (V, \cI_1)$
and $\cM_2 = (V, \cI_2)$ over the same ground set, and a common independent set 
$S\in\cI_1\cap \cI_2$, the \emph{exchange graph} $G(S)$ is a directed bipartite graph on vertex set $V\cup\{s,t\}$ with the following arcs (or directed edges):
\begin{enumerate}
    \item $(s,b)$ for $b \in \bar S$ when $S + b \in \cI_1$.
    \item $(b,t)$ for $b \in \bar S$ when $S + b \in \cI_2$.
    \item \label{itm:exch-i3} $(a,b)$ for $b \in \bar S, a \in S$ when $S + b - a \in \cI_1$.
    \item \label{itm:exch-i4} $(b,a)$ for $b \in \bar S, a \in S$ when $S + b - a \in \cI_2$.
\end{enumerate}
We will denote the set of elements at distance $k$ from $s$ by the distance-layer $D_k$.
\end{definition}

\begin{definition}[Shortest augmenting path]
A shortest $(s,t)$-path $p = (s, b_1, a_1, b_2, a_2, \ldots,\allowbreak a_{\ell}, b_{\ell+1},t)$ (with $b_i\in \bar{S}$ and $a_i\in S$) in $G(S)$ is called a \emph{shortest augmenting path}.
We can \emph{augment} $S$ along the path $p$ to obtain
$S'=S\oplus p = S+b_1-a_1+b_2-a_2\ldots+b_{\ell+1}$, which is well-known to also be a common independent set (with $|S'| = |S|+1$) \cite{cunningham1986improved}.
Conversely, there must exist a shortest augmenting path whenever $|S| < r$.
\end{definition}

The following lemma is very useful for $(1-\eps)$-approximation algorithms since it essentially says that one needs only to consider paths up to length $O(\frac{1}{\eps})$.

\begin{lemma}[Cunningham \cite{cunningham1986improved}] \label{clm:dists}
If the length of the shortest $(s,t)$-path in $G(S)$ is at least $2\ell+2$, then 
$|S| \ge (1-O(1/\ell))r$.
\end{lemma}

\begin{lemma}[Exchange discovery by binary search \cite{chakrabarty2019faster, nguyen2019note}] \label{lem:find-exchange}
Suppose
$\cM = (V,\cI)$ is a matroid, $Y\subseteq X\in \cI$, and $b\not\in X$ such that $X+b\notin\cI$.
Then, using $O(\log |Y|)$ independence queries one can find
some $a\in Y$ such that $X+b-a\in\cI$ or determine that none exist.\footnote{When $X = S$, we can use this to find edges of type \ref{itm:exch-i3} and \ref{itm:exch-i4} in the exchange graph.}
\end{lemma}

\subsection*{Augmenting Sets}
A generalization of the classical \emph{augmenting paths}---called \emph{augmenting sets}---play a key role in the approximation algorithm
of \cite{chakrabarty2019faster}, and therefore also in the modified version of this algorithm presented in this paper.
In order to efficiently find ``good'' augmenting sets, the algorithm works with
a relaxed form of them instead: \emph{partial} augmenting sets.
The following definitions and key properties of (partial) augmenting sets are copied from \cite{chakrabarty2019faster} where one can find the corresponding proofs.

\begin{definition}[{Augmenting Sets, from \cite[Definition~24]{chakrabarty2019faster}}]
 Let $S\in\cI_{1}\cap\cI_{2}$ and $G(S)$ be the corresponding exchange
graph with shortest $(s,t)$-path of length $2(\ell+1)$ and distance layers $D_1, D_2, \ldots, D_{2\ell+1}$.
A collection of sets $\Pi_{\ell}:=(B_{1},A_{1},B_{2},A_{2},\ldots,A_{\ell},B_{\ell+1})$
form an \emph{augmenting set} (of \emph{width} $w$)
in $G(S)$ if the following conditions are satisfied: %$B_{0}\subseteq D_{0},A_{1}\subseteq D_{1},B_{1}\subseteq D_{2},...,B_{l}\subseteq D_{2l}$
\begin{enumerate}[(a)]
\item For $1\leq k\leq\ell+1$, we have $A_{k}\subseteq D_{2k}$ and $B_{k}\subseteq D_{2k-1}$.
\item $|B_{1}|=|A_{1}|=|B_{2}|=\cdots=|B_{\ell+1}|=w$
\item $S+B_{1}\in\mathcal{I}_{1}$
\item $S+B_{\ell+1}\in\mathcal{I}_{2}$
\item For all $1\leq k\leq\ell$, we have $S-A_{k}+B_{k+1}\in\mathcal{I}_{1}$
\item For all $1\leq k\leq\ell$, we have $S-A_{k}+B_{k}\in\mathcal{I}_{2}$
\end{enumerate}
\end{definition}

\begin{definition}[{Partial Augmenting Sets, from \cite[Definition~37]{chakrabarty2019faster}}]
We say that $\Phi_{\ell}:=(B_{1},A_{1},B_{2},A_{2},\ldots,A_{\ell},B_{\ell+1})$
 forms a \emph{partial augmenting set} if it satisfies the conditions (a), (c), (d), and (e)
 of an \emph{augmenting set}, plus the following two relaxed conditions:
\begin{description}
\item[(b)]  \label{itm:par-b} $|B_{1}|\ge|A_{1}|\ge|B_{2}|\ge\cdots\ge|B_{\ell+1}|$. 
\item[(f)] \label{itm:par-f} For all $1\leq k\leq\ell$, we have $\rk_{2}(S-A_{k}+B_{k})=\rk_{2}(S)$.
\end{description}
\end{definition}

\begin{theorem}[{from \cite[Theorem~25]{chakrabarty2019faster}}]
\label{thm:aug-means-aug} Let $\Pi_{\ell}:=(B_{1},A_{1},B_{2},A_{2},\cdots,B_{\ell},A_{\ell},B_{\ell+1})$
be the an augmenting set in the exchange graph $G(S)$.
 Then the set $S':=S\oplus\Pi_{\ell}:=S+B_{1}-A_{1}+B_{2}-\cdots+B_{\ell}-A_{\ell}+B_{\ell+1}$ is a common independent set.\footnote{Note that $|S'| = |S|+w$, where $w$ is the width of $\Pi_\ell$. In particular, an augmenting set with width $w = 1$ is exactly an augmenting path.}
\end{theorem}

We also need the notion of \emph{maximal} augmenting sets, which naturally correspond to a maximal ordered collection of shortest augmenting paths, where, after augmentation, the $(s,t)$-distance must have increased. The following are due to \cite{chakrabarty2019faster}.

\begin{definition}[{Maximal Augmenting Sets,
from \cite[Definition~35]{chakrabarty2019faster}}]
\label{def:subset}
Let $\Pi_{\ell}=(B_{1},A_{1},B_{2},\allowbreak\cdots,\allowbreak B_{\ell},A_{\ell},B_{\ell+1})$
and $\tilde{\Pi}_{\ell}=(\tB_{1},\tA_{1},\tB_{2},\cdots,\tB_{\ell},\tA_{\ell},\tB_{\ell+1})$
be two augmenting sets in $G(S)$. We say $\tilde{\Pi}_{\ell}$ {\em contains}
$\Pi_{\ell}$
if $B_{k}\subseteq\tB_{k}$ and $A_{k}\subseteq\tA_{k}$,
for all $k$.
An augmenting set $\Pi_{\ell}$ is called \emph{maximal} if there
exists no other augmenting set $\tilde{\Pi}_{\ell}$ containing $\Pi_{\ell}$. 
\end{definition}

\begin{theorem}[from {\cite[Theorem~36]{chakrabarty2019faster}}]
\label{thm:max-inc-dist}
An augmenting set $\Pi_{\ell}$ is maximal if and only if there is no augmenting path of length at most $2(\ell+1)$ in $G(S\oplus \Pi_{\ell})$.
\end{theorem}

\section{Improved Approximation Algorithm} \label{sec:approx}

Our algorithm closely follows the algorithm of
Chakrabarty-Lee-Sidford-Singla-Wong \cite[Section 6]{chakrabarty2019faster}.
The algorithm runs in phases, where in each phase the algorithm finds
a maximal set of augmentations to perform, so that the $(s,t)$-distance in the exchange graph increases between phases. By \Cref{clm:dists}, only $O(1/\eps)$ phases are necessary.

In the beginning of a phase, the algorithm runs a breadth-first-search to compute the distance layers $D_1, D_2, \ldots D_{2\ell+1}$ in the exchange graph $G(S)$,
where $S$ is the current common independent set.
The total number of independence queries, across all phases, for these BFS's can be bounded by $O(n \log(r)/\eps)$. We refer to \cite[Algorithm~4, Lemma~19, and Proof of Theorem 21]{chakrabarty2019faster} for how to implement
such a BFS efficiently.

After the distance layers have been found, the search for a maximal augmenting set begins.
We start by summarizing on a high level how the algorithm of \cite{chakrabarty2019faster} does this in two stages:
\begin{enumerate}
\item The first stage keeps track of a \emph{partial} augmenting set which it keeps \emph{refining} by a series of operations on adjacent distance layers in the exchange graph, to make it closer to a \emph{maximal} augmenting set.

\item When we are ``close enough'' to a \emph{maximum} augmenting set, the second stage handles the last few augmenting paths---for which the first stage slows down---by finding the remaining augmenting paths individually one at a time.
\end{enumerate}
Here we give two independent improvements over the algorithm of \cite{chakrabarty2019faster}, one for each stage. The first improvement
is to replace the refine operations in the first stage by a new subroutine \texttt{RefineABA} (\cref{sec:refineaba}) working on \emph{three} consecutive layers instead of two. This allows us to measure progress in terms of $r$ instead of $n$.
The second improvement is for the second stage where we, instead of finding arbitrary augmenting paths, work directly on top of the output of the first stage and find a specific type of \emph{valid paths} 
with respect to the partial augmenting set, using a new
a subroutine \texttt{RefinePath} (\cref{sec:refinepath}).

\subsection{Implementing a Phase: Refining}

The basic refining ideas and procedures in this section are the same as in \cite{chakrabarty2019faster}. The goal is to keep track of
a partial augmenting set $\Phi_\ell = (B_1, A_1, B_2, \ldots, A_{\ell}, B_{\ell+1})$
which is iteratively made ``better'' through some \emph{refine procedures}.
Eventually, the partial augmenting set will become a maximal augmenting set, which concludes the phase. Towards this goal,  we maintain three types of elements in each layer:

\begin{description}
\item[Selected.] Denoted by $A_k$ or $B_k$. These form the partial augmenting set $\Phi_\ell = (B_1, A_1, B_2,\allowbreak\ldots,\allowbreak A_{\ell}, B_{\ell+1})$.
\item[Removed.] Denoted by $R_k$. These elements are safe to disregard from further computation (i.e. deemed useless) when refining $\Phi_\ell$ towards a maximal augmenting set.
\item[Fresh.] Denoted by $F_k$. These are the elements that are neither \emph{selected} nor \emph{removed}.
\end{description}

\begin{figure}[ht]
\begin{center}
\includegraphics[width=9cm]{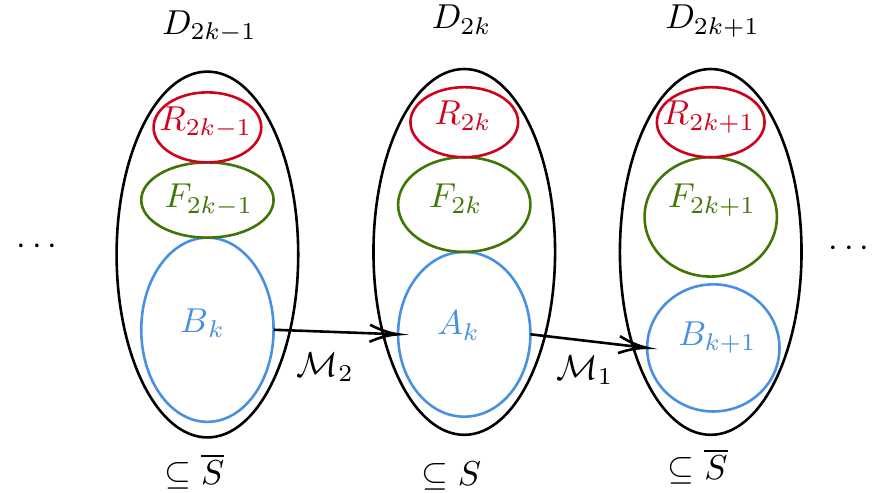}
\end{center}
\caption{
An illustration of a few neighboring layers.
Note that $(B_k, R_{2k-1}, F_{2k-1})$ form a partition of odd layer $D_{2k-1}\subseteq \bar{S}$, and $(A_k, R_{2k}, F_{2k})$ form a partition of even layer $D_{2k}\subseteq S$.
}
\end{figure}

Elements can change their types from \emph{fresh} $\to $ \emph{selected} $\to$
\emph{removed}, but never in the other direction.
Initially, we start with all elements being fresh.\footnote{This differs slightly from \cite{chakrabarty2019faster}, where the initially $B_1$ is greedily picked to be maximal so that $S+B_1\in \cI_1$, while the rest of the elements are fresh.}
For convenience, we also define ``imaginary'' layers $D_0$ and $D_{2\ell+2}$
with $A_{0} = R_{0} = F_{0} = D_{0} =  A_{\ell+1} = R_{2\ell+2} = F_{2\ell+2} = D_{2\ell+2} = \emptyset$.
The algorithm maintains the following \emph{phase invariants} (which are initially satisfied) during the refinement process:

\begin{definition}[{Phase Invariants, from \cite[Section~6.3.2]{chakrabarty2019faster}}]
The \emph{phase invariants} are:
\begin{description}
\item[(a-b)] $\Phi_\ell = (B_1, A_1, B_2, \ldots, A_\ell, B_{\ell+1})$ forms a partial augmenting set.\footnote{The naming of this invariant as (a-b) is to be consistent with \cite{chakrabarty2019faster} where this condition is split up into two separate items (a) and (b).}

\item[(c)] For $1\leq k\leq\ell$, for any $X\subseteq B_{k+1}+F_{2k+1}=D_{2k+1}-R_{2k+1}$,
if $S-(A_{k}+R_{2k})+X\in\mathcal{I}_{1}$ then $S-A_{k}+X\in\mathcal{I}_{1}$. 
\footnote{An equivalent condition for (c) is:
$\rk_1(W - R_{2k}) = \rk_1(W)-|R_{2k}|$, where $W = S-A_{k} + (D_{2k+1}-R_{2k+1})$.
}

\item[(d)]  $ \rk_2(W + R_{2k-1}) = \rk_2(W)$ where $W = S-(D_{2k}-R_{2k})+B_{k}$.
\end{description}
\end{definition}

\begin{remark}
Invariant (c) essentially says that
if $R_{2k+1}$ is ``useless'', then so is $R_{2k}$.
Similarly, Invariant (d) says that if $R_{2k}$
is ``useless'', then so is $R_{2k-1}$.
Together they imply that we can safely ignore all the removed elements.
\end{remark}

\begin{lemma} \label{lem:invariants-maximality}
Suppose that (i) the phase invariants hold; (ii) $|B_1| = |A_1| = \cdots = |B_{\ell+1}|$;
and (iii) $B_1$ is a maximal subset of $D_1\setminus R_1$ satisfying $S+B_1\in \cI_1$.
Then $(B_1, A_1, \ldots, B_{\ell+1})$ is a maximal augmenting set.
\end{lemma}
\begin{proof}[Proof idea.]
(See \cite[Proof of Lemma~44]{chakrabarty2019faster} for a complete proof).
If it was not maximal, there exists an augmenting path
$(b_1, a_1, \ldots, b_{\ell+1})$ in the exchange graph after
augmenting along $(B_1, A_1, \ldots, B_{\ell+1})$.
However, (iii) then says that $b_1$ must have been removed since it cannot be fresh.
But if $b_1$ is removed, then so was $a_1$, then so was $b_2$ etc., by invariants (c)
and (d) (this requires a technical, but straightforward, argument). However, $b_{\ell+1}$ cannot have been removed (by invariant (d)),
which gives the desired contradiction.
\end{proof}

\subsubsection{Refining Two Adjacent Layers}
We now present the basic refinement procedures from \cite{chakrabarty2019faster},
which
are
operations on neighboring layers.
There is some asymmetry in how (odd, even) and (even, odd) layer-pairs 
are handled, arising from the inherent asymmetry of the independence query between $S$
and $\bar{S}$, but the ideas are the same.
\begin{description}
\item[$\texttt{RefineAB}(k)$] extends $B_{k+1}$ as much as possible while
respecting invariant (a-b) (Lines 1-2).
Then a maximal collection of element in $A_{k}$ which can be ``matched'' to $B_{k+1}$ is found, and the others elements in $A_{k}$ are removed (Lines 3-4).

\item[$\texttt{RefineBA}(k)$] finds a maximal subset $B_{k}$ that can be ``matched''
to $A_k+F_{2k}$, and removes the other elements of $B_k$ (Lines 1-2).
Then $A_k$ is extended with elements from $F_{2k}$ which are the endpoints of the above ``matching'' (Lines 3-4).
\end{description}

\begin{algorithm}[H]
\caption{$\texttt{RefineAB}(k)$\hfill(called $\texttt{Refine1}$ in \cite[Algorithm~9]{chakrabarty2019faster})}
\begin{algorithmic}[1]

\State
Find maximal $B\subseteq F_{2k+1}$ s.t. $S-A_{k}+B_{k+1}+B\in\mathcal{I}_{1}$

\State
$B_{k+1}\longleftarrow B_{k+1}+B,F_{2k+1}\longleftarrow F_{2k+1}-B$

\State
Find maximal $A\subseteq A_{k}$ s.t. $S-A_{k}+B_{k+1}+A\in\mathcal{I}_{1}$

\State
$A_{k}\longleftarrow A_{k}-A,R_{2k}\longleftarrow R_{2k}+A$

\end{algorithmic}
\end{algorithm}

\begin{algorithm}[ht]
\caption{$\texttt{RefineBA}(k)$\hfill(called $\texttt{Refine2}$ in \cite[Algorithm~10]{chakrabarty2019faster})}
\begin{algorithmic}[1]

\State
Find maximal $B\subseteq B_{k}$ s.t. $S-(D_{2k}-R_{2k})+B\in\mathcal{I}_{2}$

\State
$R_{2k-1}\longleftarrow R_{2k-1}+B_{k}\backslash B,B_{k}\longleftarrow B$

\State
Find maximal $A\subseteq F_{2k}$ s.t. $S-\left(D_{2k}-R_{2k}\right)+B_{k}+A\in\mathcal{I}_{2}$

\State
$A_{k}\longleftarrow A_{k}+F_{2k}\backslash A,F_{2k}\longleftarrow A$

\end{algorithmic}
\end{algorithm}

The following properties of the \texttt{RefineAB} and \texttt{RefineBA} methods
are proven in  \cite{chakrabarty2019faster}.

\begin{lemma}[{from \cite[Lemmas~40-42]{chakrabarty2019faster}}]
\label{lemma:refine-invariants}
Both $\texttt{RefineAB}$
and $\texttt{RefineBA}$ preserve the invariants.
Also: after $\texttt{RefineAB}(k)$ is run, we have
$|A_{k}|=|B_{k+1}|$ (unless $k=0$). After $\texttt{RefineBA}(k)$ is run, we have
$|B_{k}|=|A_{k}|$ (unless $k=\ell+1$). 
\end{lemma}
\begin{lemma}[{from \cite[Lemma~45]{chakrabarty2019faster}}]
\label{lemma:refine-queries}
$\texttt{RefineAB}$ can be implemented with $O(|D_{2k}|+|D_{2k+1}|)$ queries.
$\texttt{RefineBA}$ can be implemented with $O(|D_{2k-1}|+|D_{2k}|)$ queries.
\end{lemma}

\begin{observation} \label{obs:refine-old}
\cref{lemma:refine-invariants} is particularly interesting. It says that
at least $|A^{old}_{k}|-|B^{old}_{k+1}|$ (respectively
$|B^{old}_{k}|-|A^{old}_{k}|$)
elements change type when running $\texttt{RefineAB}$ (respectively $\texttt{RefineAB}$).
\end{observation}

\begin{remark}
\Cref{obs:refine-old} is used in \cite{chakrabarty2019faster} to bound the number of times one needs to refine the partial augmenting set. Indeed, every element can only change its type a constant number of times.
In a single refinement pass, procedures \texttt{RefineAB}(k) and \texttt{RefineBA}(k)
are called for all $k$, and we obtain a telescoping sum guaranteeing us that $|B^{old}_1| - |B^{old}_{\ell+1}|$ elements have changed their types. Hence, after $O(\sqrt{n})$ refinement passes we have $|B_1| -|B_{\ell+1}| \le \sqrt{n}$, and we are    ``close'' to having a maximal augmenting set---only around $\sqrt{n}$ many augmenting paths away.
This is essentially what lets \cite{chakrabarty2019faster} obtain
their subquadratic $\tO(n^{1.5}/\text{poly}(\eps))$ algorithm.
\end{remark}

\subsubsection{Refining Three Adjacent Layers}
\label{sec:refineaba}

We are now ready to present the new $\texttt{RefineABA}$ method (\cref{alg:refineABA}),
which is \textbf{not} present in \cite{chakrabarty2019faster}.
This method works similarly to \texttt{RefineAB} and \texttt{RefineBA}, but on \textbf{three} (instead of two) consecutive layers $(D_{2k}, D_{2k+1}, D_{2k+2})$
with the corresponding sets $(A_{k}, B_{k+1}, A_{k+1})$.

The motivation for this new procedure is that we can get a stronger version of \cref{obs:refine-old}:
after running $\texttt{RefineABA}(k)$ we want that at least $|A^{old}_{k}|-|A^{old}_{k+1}|$ element in \textbf{even} layers have changed types.
Note that there are at most $|S| \le r$ elements in the even layers (as opposed to $n$
elements in total, which can be much larger), so this means we need to refine
the partial augmenting set fewer times when using $\texttt{RefineABA}$
compared to when just using $\texttt{RefineAB}$ and $\texttt{RefineBA}$.
In particular, we will get that after $O(\sqrt{r})$ refinement passes, $|B_1|-|B_{\ell+1}|\le \sqrt{r}$.

\begin{remark}
A natural question to ask is if it actually could be the case that
only elements in odd layers (i.e.\ those in $\bar{S}$ which there are up to $n$ many of)
change their type (while elements in even layers do not) during the refinement passes in the algorithm of \cite{chakrabarty2019faster} (which only uses the two-layer refinement procedures)?
That is, is the new three-layer refinement procedure necessary?
The answer is yes.
Consider for example the case with 5 layers
$B_1\subseteq D_1; A_1\subseteq D_2; B_2\subseteq D_3; A_2 \subseteq D_4; B_3\subseteq D_5$ where $q := |B_1| = |A_1|$ and $|A_2| = |B_3| = 0$.
Refining the consecutive pair $(B_1, A_1)$ or $(A_2, B_3)$ will not do anything.
When refining $(A_1, B_2)$ it could be the case that only $B_2$ increases (say
any $q$-size subset in $D_3$ can be ``matched'' with $A_1$).
Similarly, when refining $(B_2,A_2)$ it could be the case that only $B_2$ decreases (say there is only a single element in $D_3$ which could be ``matched'' with anything in the next layer $D_4$, then it is unlikely that this specific element is already selected in $B_2$). In this case, we would need to run the two-layer refinement procedures around
$|D_3|/q \approx n/q$ times before anything other than $B_2$ changes.
In contrast, the new $\texttt{RefineABA}$ method would, when run on $(A_1, B_2, A_2)$,
    terminate with $|A_1| = |B_2| = |A_2|$ (that is it would have found the ``special'' element in $D_3$ the first time it is run).
\end{remark}

\begin{figure}[ht]
\begin{center}
\includegraphics[width=10cm]{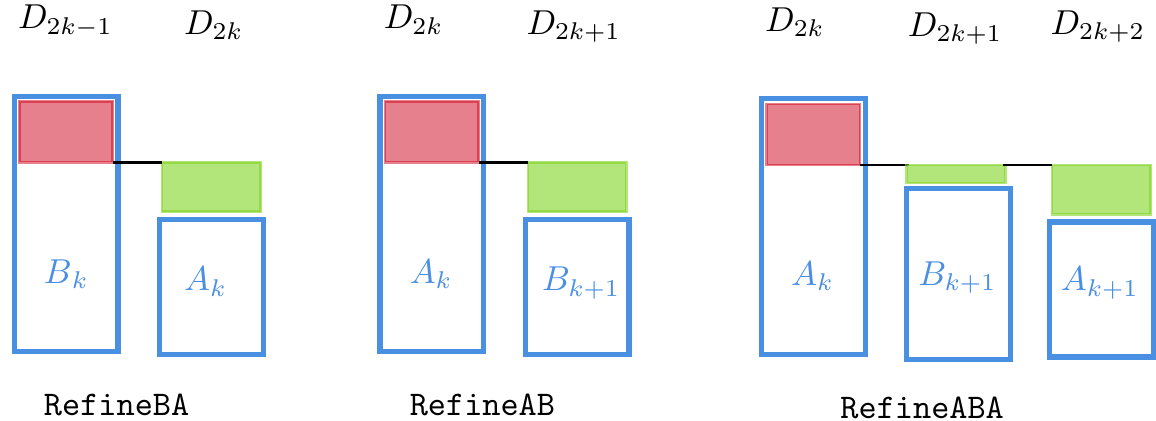}
\end{center}
\caption{
An illustration how the different refine methods change the partial augmenting sets.
Newly selected elements are marked in green, while newly removed elements are marked in red.
}
\end{figure}

To explain how $\texttt{RefineABA}$ works, let us start with a simple case, namely when $S = \emptyset$, i.e.\ there is only one layer between $s$ and $t$ in the exchange graph. Here, finding a maximal augmenting set is the same as finding some maximal set $B$ which is independent in both matroids. Running $\texttt{RefineAB}$ would extend this $B$
with elements as long as it is independent in the first matroid (ignoring the second matroid), while $\texttt{RefineBA}$ would throw away elements from $B$ until it is independent in the second matroid (now ignoring the first matroid). If we just alternate running $\texttt{RefineAB}$ and
$\texttt{RefineBA}$ we would in the worst case need to do this up to $n$ times (which is too expensive).
Instead, there is a very simple greedy algorithm that efficiently
finds a maximal set $B$ independent in both of the matroids\footnote{This algorithm on its own is a well-known $\frac{1}{2}$-approximation algorithm for matroid intersection.}: \emph{for each element, include it in $B$ if this does not break independence for either matroid}. This is akin to how our
$\texttt{RefineABA}$ method works: it looks at the constraints from both matroids simultaneously (both neighboring layers) and greedily selects $B$.

In the general case, $\texttt{RefineABA}$ can be seen as running $\texttt{RefineAB}$ and $\texttt{RefineBA}$
simultaneously.
The algorithm starts by asserting $|B_{k+1}|=|A_{k+1}|$ (so that $S+B_{k+1}-A_{k+1}\in \cI_2$) by running \texttt{RefineBA}.
So now we have both $S+B_{k+1}-A_{k}\in \cI_1$ and $S+B_{k+1}-A_{k+1}\in \cI_2$, and the algorithm proceeds to greedily 
extend $B_{k+1}$
while it is still consistent with both the previous layer $A_k$ and the next layer $A_{k+1} + F_{2k+2}$.
Some care has to be taken here to also mark elements as removed to preserve the phase invariants.
Finally,
the algorithm decreases the size of $A_k$, respectively increases the size of $A_{k+1}$, to both match $|B_{k+1}|$.

\begin{algorithm}[ht]
\caption{$\texttt{RefineABA}(k)$}
\label{alg:refineABA}
\begin{algorithmic}[1]

\State $\texttt{RefineBA}(k+1)$
\For{$x \in F_{2k+1}$}
\If{$S-A_k + B_{k+1} + x \in \cI_1$}
\If{$S-A_{k+1} - F_{2k+2} + B_{k+1} + x \in \cI_2$}
\State $B_{k+1} \gets B_{k+1} + x, \quad F_{2k+1} \gets F_{2k+1} - x$
\Comment{Select $x$}
\Else
\State $R_{2k+1} \gets R_{2k+1} + x, \quad F_{2k+1} \gets F_{2k+1} - x$
\label{lst:line:remove-x}
\Comment{Remove $x$}
\EndIf
\EndIf
\EndFor
\State $\texttt{RefineBA}(k+1)$
\State $\texttt{RefineAB}(k)$
\end{algorithmic}
\end{algorithm}

We now state some properties of \texttt{RefineABA}.
These properties are relatively straightforward---although technical and notation-heavy---to prove.

\begin{lemma} \label{lem:aba-invariants}
\texttt{RefineABA}$(k)$ preserves the phase invariants.
\end{lemma}

\begin{lemma}
\label{lem:aba-sizes}
After \texttt{RefineABA}$(k)$ is run, we have $|A_k| = |B_{k+1}| = |A_{k+1}|$
(unless $k = 0$ or $k=\ell$, where the sets $A_0 = A_{\ell+1} = \emptyset$ are ``imaginiary'').
\end{lemma}

\begin{lemma}
\label{lem:aba-queries}
\texttt{RefineABA}$(k)$ uses $O(|D_{2k}|+|D_{2k+1}|+|D_{2k+2}|)$ independence queries.
\end{lemma}

\begin{proof}[Proof of \cref{lem:aba-invariants}]
Intuitively, the only tricky part is showing that invariant (c) is preserved when some $x$ is removed in line~\ref{lst:line:remove-x}.
We can pretend that we add $x$ to $B_{k+1}$ temporarily, and then run $\texttt{RefineBA}(k+1)$ in a way which would remove this $x$ immediately (and thus removing $x$ did indeed preserve the invariants). We present a formal proof below.

We already know that $\texttt{RefineAB}$ and
$\texttt{RefineBA}$ preserve the invariants by \cref{lemma:refine-invariants},
so it suffices to check that the for-loop starting in line~2 preserves the invariants.
We verify that this is the case after processing each $x\in F_{2k+1}$ in the for-loop:
\begin{description}
    \item[Invariant (a-b)] holds by design: when $x$ is added to $B_{k+1}$ we know both
that $S-A_{k}+B_{k+1}+x\in \cI_1$ and $\rk_2(S-A_{k+1}+B_{k+1})$ cannot decrease.
Note also that $\rk_2(S-A_{k+1}+B_{k+1})\le \rk_2(S)$ when $k+1\le \ell$ too (so it cannot increase either), since
otherwise there must exist
some $b\in B_{k+1}$ so that $S+b \in \cI_2$ (by the matroid exchange property) which is impossible since we are not in the last layer (the layer preceding $t$ in $G(S)$).

    \item[Invariant (c)] trivially holds, since the set $B_{k+1}+F_{2k+1}$ will only decrease, which only restricts the choice of $X\subseteq B_{k+1}+F_{2k+1}$.

    \item[Invariant (d)] will also be preserved.
    We need to argue that this is the case when $x$ is removed in line~7.
    Let $W := S-A_{k+1}-F_{2k+2}+B_{k+1} = S-(D_{2k+2} - R_{2k+2}) + B_{k+1}$,
    and $R^{old}_{2k+1}$ be the set $R_{2k+1}$ before $x$ was added to it.
    First note that $W\in \cI_2$, since this holds after the \texttt{RefineBA} call in line~1, (since $|A_{k+1}| = |B_{k+1}|$ after this call)
    and $B_{k+1}$ is only extended with elements which preserve this property.
    This means that $\rk_2(W+x) = \rk_2(W) = |W|$, since
    $W +x = S - A_{k+1} - F_{2k+2} + B_{k+1}+x \notin \cI_2$.
    Since the invariant held before, we also know that
    $\rk_2(W+R^{old}_{2k+1}) = \rk_2(W) = |W|$.
    Hence $W$ is a maximal independent (in $\cM_2$) subset
    of $W+R^{old}_{2k+1}+x$, as neither $x$ nor elements from $R^{old}_{2k+1}$
    can be used to extend it.
    Hence $\rk_2(W+R^{old}_{2k+1}+x) = |W| = \rk_{2}(W)$; that is invariant (d) is preserved. \qedhere{}
\end{description}
\end{proof}

\begin{proof}[Proof of \cref{lem:aba-sizes}]
We focus our attention on the $\texttt{RefineBA}$ and
$\texttt{RefineAB}$ calls in lines~8-9, and argue that they do not change $B_{k+1}$.
This would prove the lemma, since by \cref{lemma:refine-invariants}
we would then have $|A_k| = |B_{k+1}|$ and $|B_{k+1}| = |A_{k+1}|$.

Indeed, 
$\texttt{RefineBA}(k+1)$ finds a maximal $B\subseteq B_{k+1}$ such that
$S-(D_{2k+2} - R_{2k+2}) + B \subseteq \cI_2$, and remove all elements not in $B$ from $B_{k+1}$.
Here, $B = B_{k+1}$ will be found, since
$S-(D_{2k+2} - R_{2k+2}) + B_{k+1} \in \cI_2$
after the for-loop in line~2 of \texttt{RefineABA}.

Similarly, we see that
$\texttt{RefineAB}(k)$ finds a maximal $B\subseteq F_{2k+1}$ such that
$S-A_{k} + B_{k+1} + B \in \cI_1$, and extend $B_{k+1}$ with this $B$.
However, only $B=\emptyset$ works, since each $x\in F_{2k+1}$
for which $S-A_k + B_{k+1} + x\in \cI_1$ was either selected or removed in lines~5 or 7.
\end{proof}

\begin{proof}[Proof of \cref{lem:aba-queries}]
\texttt{RefineAB}$(k)$ uses $O(|D_{2k}|+|D_{2k+1}|)$
queries,
and \texttt{RefineBA}$(k+1)$ uses $O(|D_{2k+1}|+|D_{2k+2}|)$ queries.
The for-loop in line~2 will use $O(|D_{2k+1}|)$ queries.
\end{proof}

\subsubsection{Refinement Pass}

We can now present the full \texttt{Refine} method (\cref{alg:refine}), which simply scans over the layers
and calls \texttt{RefineABA} on them.
Our \texttt{Refine}
is a modified version of \texttt{Refine} from \cite[Algorithm~11]{chakrabarty2019faster} using our new
\texttt{RefineABA} method instead of just \texttt{RefineAB} and \texttt{RefineBA}.
Just replacing the \texttt{Refine} method in the final algorithm of \cite{chakrabarty2019faster} with our modified \texttt{Refine} below leads to an $\tO(n\sqrt{r} / \eps^{1.5})$-query algorithm
(compared to their $\tO(n^{1.5} / \eps^{1.5})$), and concludes our first improvement (as discussed in \cref{itm:approx-improvement-r} in \cref{sec:overview}).

\begin{algorithm}[ht]
\caption{$\texttt{Refine}(k)$}
\label{alg:refine}
\begin{algorithmic}[1]
\For{$k = \ell,\ \ell-1,\ \ell-2,\ \ldots,\ 1,\ 0$}
\State \texttt{RefineABA}(k)
\EndFor
\end{algorithmic}
\end{algorithm}

The following \cref{lem:refine-types} will be useful to bound the number of $\texttt{Refine}$ calls needed in our final algorithm, and closely corresponds to \cite[Corollary~43]{chakrabarty2019faster}. Our \texttt{Refine} implementation has the advantage that
it only counts the elements in the \textbf{even} layers,
of which there are at most $r$.

\begin{lemma}
\label{lem:refine-types}
Let $(B^{old}_1, A^{old}_1, \ldots)$
and $(B^{new}_1, A^{new}_1, \ldots)$ be the sets before and after \texttt{Refine} is run.
Then at least $|B^{new}_1| - |B^{new}_{\ell+1}|$ elements in even layers
have changed types.
\end{lemma}
\begin{proof}

Note that whenever $A_k$ changes, it is because some elements changed it types in $D_{2k}$.
In particular, if the size of $A_k$ increases (respectively decreases) by $z$, at least
$z$ elements will change types from fresh to selected 
(respectively from selected to removed) in $D_{2k}$.

After the first iteration $|A_{\ell}| = |B^{new}_{\ell+1}|$, so
at least $|A^{old}_{\ell}| - |B^{new}_{\ell+1}|$ elements in $D_{2\ell}$ changed types.
Similarly, after the iteration when $k = i$ (for $1\le i \le \ell-1$),
$|A_{i}| = |A_{i+1}|$, and hence at least $|A^{old}_{i}| - |A_{i}|$ elements in $D_{2i}$ changed types plus at least $|A_{i+1}| - |A^{old}_{i+1}|$  elements in $D_{2i+2}$ changed types.\footnote{$|A_{i+1}| \le |A^{old}_{i+1}|$ just before the
$\texttt{RefineABA}(i)$ call, since earlier iterations can only have decreased the size of $|A_{i+1}|$.}
Finally, after the last iteration  $|A_{1}| = |B^{new}_{1}|$,
and hence at least $|B^{new}_{1}| - |A^{old}_{1}|$ elements in $D_{2}$ changed types.

The above terms telescope, and we conclude that at least $|B^{new}_1| - |B^{new}_{\ell+1}|$ elements in the even layers changed its types when \texttt{Refine}
was run.
\end{proof}

\begin{lemma}
\label{lem:refine-queries}
\texttt{Refine} uses $O(n)$ independence queries.
\end{lemma}
\begin{proof}
This follows directly by \cref{lem:aba-queries}.
\end{proof}

\subsection{Refining Along a Path} \label{sec:refinepath}
If we just run $\texttt{Refine}$ until we get a maximal augment set (i.e.\ until $|B_1| = |B_{\ell+1}|$) we need to potentially run $\texttt{Refine}$ as many as $\Theta(r)$ times, which needs too many independence queries.
\cref{lem:refine-types} tells us that \texttt{Refine} makes the most ``progress''
while $|B_1| - |B_{\ell+1}|$ is large: in fact, only $O(r/p)$ calls to \texttt{Refine}
is needed until $|B_1| - |B_{\ell+1}| \le p$.
The idea in \cite{chakrabarty2019faster} is thus to stop refining
when $|B_1| - |B_{\ell+1}|$ is small enough and fall back to finding augmenting paths one at a time (they prove that one needs to find at most $O(
(|B_1| - |B_{\ell+1}|) \ell)$ many). We use a similar idea in that we
swap to a different procedure when $|B_1| - |B_{\ell+1}|$ is small enough, the difference being that we still work with the partial augmenting set. This will let us
show that only $O(|B_1| - |B_{\ell+1}|)$ many ``paths'' need to be found, saving a factor $\ell \approx \frac{1}{\eps}$ compared to \cite{chakrabarty2019faster}.

This section thus describes the second improvement (as discussed in \cref{itm:approx-improvement-eps} in \cref{sec:overview}). Note that this improvement is independent of the first improvement (i.e.\ the three-layer refine).
We aim to prove the following lemma.

\begin{lemma}
\label{lem:refine-path}
There exists a procedure (\texttt{RefinePath}, \cref{alg:refinepath}), which uses
$O(n \log r)$ independence queries, preserves the invariants, and either:
\begin{enumerate}[i.]
\item Increases the size of $B_{\ell+1}$ by at least $1$.
\item Terminates with $(B_1, A_1, \ldots, B_{\ell+1})$ being a maximal augmenting set.
\end{enumerate}
\end{lemma}

\texttt{RefinePath} attempts to find what we call a \emph{valid path}.
What we want is a sequence of elements which we can add to the partial augmenting set
without violating the invariants and the properties of the partial augmenting set.
It turns out (not very surprisingly) that such sequences of elements can be characterized by a notion of \emph{paths} in something which resembles the \emph{exchange graph with respect to our partial augmenting set}.
This is what motivates the definition of \emph{valid paths} below.

\begin{figure}[ht]
\begin{center}
\includegraphics[width=9.4cm]{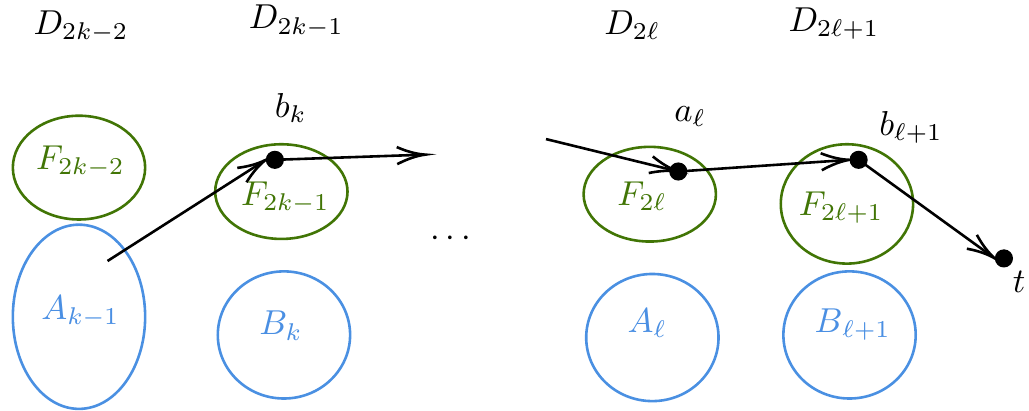}
\end{center}
\caption{
A valid path $(b_k, \ldots, a_{\ell}, b_{\ell+1}, t)$ ``starting'' from
the partial augmenting set at $A_{k-1}$, so that we can use \cref{lem:valid-path} and augment along it.
}
\end{figure}

\begin{definition}[Valid path]
A sequence $(b_i,a_i,b_{i+1}, \ldots, b_{\ell+1}, t)$
(or $(a_i, b_{i+1}, \ldots, b_{\ell+1}, t)$)
is called a \emph{valid path} (with respect to the partial augmenting set)
if for all $k \ge i$:
\begin{enumerate}[(a)]
\item $a_k\in F_{2k}$ and $b_k\in F_{2k-1}$.
\item $S + B_{\ell+1} + b_{\ell+1} \in \cI_2$.
\item $S-A_k + B_k - a_k + b_k \in \cI_2$.
\item $S-A_k + B_{k+1} - a_k + b_{k+1} \in \cI_1$.
\end{enumerate}
\end{definition}

\begin{remark}
Compare the properties of valid paths with the edges in the exchange graph from \cref{def:exchange}.
A valid path is essentially a path in the exchange graph
\emph{after} we have already augmented $S$ by our partial augmenting set (even though
this exchange graph is not exactly defined, since it is not guaranteed that $S$ remains a common independent set when augmented by a \emph{partial} augmenting set).
\end{remark}

\begin{lemma} \label{lem:valid-path}
If $p = (b_i,a_i,b_{i+1}, \ldots, b_{\ell+1},t)$ is a valid path starting at $b_i$, 
such that $S-A_{i-1} + B_{i} + b_i \in \cI_1$, then
$(B_1, A_1, \ldots, B_{i-1}, A_{i-1},\ B_i+b_i, A_i+a_i, \ldots, B_{\ell+1}+b_\ell)$
is a partial augmenting set satisfying the invariants.
\end{lemma}
\begin{proof}
That it forms a partial augmenting set is true by the definition of valid paths,
and the fact that $S-A_{i-1}+B_i+b_i\in \cI_1$.
Indeed, it cannot be the case that $|A_{i-1}| < |B_i+b_i|$ when $i>1$, since then
$\rk_1(S-A_{i-1}+B_i+b_i) > |S| = \rk_1(S)$ implies that some element $x\in (B_i+b_i)$ satisfies $S+x \in \cI_1$ (i.e.\ it is in the first layer $D_1$) by the exchange property of matroids.
Invariants (c) and (d) are trivially true since the sets $A_k$ and $B_k$ are only extended.
\end{proof}

The goal of $\texttt{RefinePath}$ (\cref{alg:refinepath}) is thus to find a valid path satisfying the conditions in \cref{lem:valid-path}.
Towards this goal, \texttt{RefinePath} will start from the last layer $D_{2\ell+1}$ and
``scan left'' in a breadth-first-search manner while keeping track of valid paths starting at each fresh vertex $x$ (the next element on such a path will be stored as $\texttt{next}[x]$). If at some point one valid path can ``enter'' the partial augmenting set in a layer, we are done and can use \cref{lem:valid-path}.
We also show that it is safe (i.e.\ preserves the invariants) to remove all the fresh elements $x$ for which we cannot find a valid path starting at $x$.

To efficiently find the ``edges'' during our breadth-first-search using only independence-queries, we use the binary-search trick from \cref{lem:find-exchange}.
However, this relies on the partial augmenting set being locally ``flat'' in the layers we are currently exploring, i.e.
$|B_k| = |A_k|$ respectively $|B_k| = |A_{k+1}|$. We can ensure this by
running $\texttt{RefineAB}$ respectively $\texttt{RefineBA}$
while performing the scan.

Now we are ready to present the pseudo-code of the \texttt{RefinePath} method (\cref{alg:refinepath}). 
Due to the asymmetry between even/odd layers and independence queries, we need to handle moving from layer $B$ to $A$ and from $A$ to $B$ a bit differently, but the ideas are similar.

\begin{algorithm}[ht]
\caption{\texttt{RefinePath}} \label{alg:refinepath}
\begin{algorithmic}[1]

\For{$k = \ell+1,\ \ell,\ \ldots,\ 2,\ 1$}

      \Comment{Process $(B_{k}, A_{k})$}
    
\State $\texttt{RefineBA}(k)$
\If{some element $a$ was added to $A_k$ in the above refine-call}
\State Add the valid path starting at $\texttt{next}[a]$ to the partial augmenting set
\label{lst:valid-path-1}
\State \Return
\EndIf
\For{each element $b\in F_{2k-1}$}
\If{$S-A_{k}-F_{2k}+B_{k} + b \notin \cI_2$}
\label{lst:check-1}
\State Remove $b$, that is: $F_{2k-1} \gets F_{2k-1}-b, \quad R_{2k-1} \gets R_{2k-1}+b$
\label{lst:remove-b}
\Else
\State \label{lst:find-1} Find an $a\in F_{2k}$ such that $S-A_{k}+B_{k} + b-a \in \cI_2$.
Let $\texttt{next}[b] = a$.
\State (Or, if $k=\ell+1$, just let $\texttt{next[b]} = t$)
\EndIf
\EndFor

\Statex
      \Comment{Process $(A_{k-1}, B_{k})$}
    
\If{some element $b\in F_{2k-1}$ satisfies $S-A_{k-1} + B_{k} + b \in \cI_1$}
\label{lst:check-2}
\State Add the valid path starting at $b$ to the partial augmenting set.
\label{lst:valid-path-2}
\State \Return
\EndIf
\State $\texttt{RefineAB}(k-1)$
\State $Q \gets F_{2k-2}$.
\For{each element $b\in F_{2k-1}$}
\While{can find $a\in Q$ such that $S-A_{k-1}+B_{k} + b-a \in \cI_1$}
\label{lst:find-2}
\State $Q \gets Q-a$. Let $\texttt{next}[a] = b$.
\EndWhile
\EndFor
\State Remove all elements in $Q$, that is: $F_{2k-2} \gets F_{2k-2}-Q,\quad
R_{2k-2} \gets R_{2k-2}+Q$.
\label{lst:remove-a}

\EndFor

\Statex
\State If we reached here, $(B_1, A_1, \ldots, B_{\ell+1})$ is a maximal augmenting set.

\end{algorithmic}
\end{algorithm}

\begin{restatable}{lemma}{PATHinv} \label{lem:refine-path-inv}
\texttt{RefinePath} preserves the invariants.
\end{restatable}
\begin{proof}
The proof is relatively straightforward, but technical.
The only non-trivial part is showing that invariants (c) and (d) are preserved
after we remove something in line~\ref{lst:remove-b} or line~\ref{lst:remove-a}. Intuitively, if we remove $b$ in line~\ref{lst:remove-b}, we can instead think of temporarily adding $b$ to $B_{k}$ and running $\texttt{RefineBA}(k)$ in such a way so that $b$ is immediately removed. A similar intuitive argument works for line~\ref{lst:remove-a}. We next present a formal proof.

We know that \texttt{RefineAB} and \texttt{RefineBA} preserve the invariants, by \cref{lemma:refine-invariants}.
We also know by \cref{lem:valid-path} that adding a valid path to the partial augmenting set also preserves the invariants.
So what remains is to show that the invariants are preserved after:
\begin{description}

\item[Line~\ref{lst:remove-b}.]
We only need to check invariant (d), the other ones trivially hold.
Let $W = S-A_k-F_{2k}+B_k = S-(D_{2k}-R_{2k}) + B_{k}$
and $R^{old}_{2k-1}$ be $R_{2k-1}$ before $b$  was added to it.
Note that $b$ is such that $W + b \notin \cI_2$,
and we know that $W \subseteq S-A_k+B_k\in \cI_2$
and hence $\rk_2(W+R^{old}_{2k-1}) = \rk_2(W) = |W|$
and $\rk_2(W+b) = \rk_2(W) = |W|$.
We thus need to show that $\rk_2(W+R^{old}_{2k-1}+b) = |W|$ too,
which is clear since $W$ is a maximal independent subset of
$W+R^{old}_{2k-1}+b$ (it can neither be extended with elements from $R^{old}_{2k-1}$
nor with $b$).

\item[Line~\ref{lst:remove-a}.]
We only need to check invariant (c), the other ones trivially hold.
We imagine we add the $a\in Q$ to $R_{2k-2}$ one-by-one, and show that the invariant (c)
is preserved after each such addition.
So consider some $a\in Q$ which will be removed,
and let $R^{old}_{2k-2}$ be the set $R_{2k-2}$ just before we added $a$ to it.
First note that $\rk_1(S-A_{k-1}+B_k + F_{2k-1} - a) = \rk_1(S-A_{k-1}+B_k + F_{2k-1}) - 1
= |S-A_{k-1}+B_k| - 1$,
as otherwise there must exist some $b\in F_{2k-1}$ such that
$S-A_{k-1}+B_k + b- a\in \cI_1$ (by the matroid exchange property), and $a$ would have been discovered in line~\ref{lst:find-2} and therefore been removed from $Q$.
So the ``return'' of adding $a$ to $S-A_{k-1}+B_k + F_{2k-1} - a$ is increasing the rank by $1$.
Now consider some arbitrary $X\subseteq B_k + F_{2k-1}$ such that
$S-A_{k-1}+X-R^{old}_{2k-2}-a \in \cI_1$.
We need to show that $S-A_{k-1}+X \in \cI_1$.
Note that $S-A_{k-1}+X-R^{old}_{2k-2}-a\subseteq S-A_{k-1}+B_{k}+F_{2k-1}-a$.
Hence, by the diminishing returns (of adding $a$) we know
$\rk_1(S-A_{k-1}+X-R^{old}_{2k-2}) \ge \rk_1(S-A_{k-1}+X-R^{old}_{2k-2}-a) + 1
= |S-A_{k-1}+X-R^{old}_{2k-2}|$,
or equivalently that $S-A_{k-1}+X-R^{old}_{2k-2} \in \cI_1$.
Since the invariant held before, we conclude that $S-A_{k-1}+X \in \cI_1$ too, which finishes the proof. \qedhere{}
\end{description}
\end{proof}

\paragraph{Valid paths.}
The algorithm keeps track of a valid path starting at each fresh vertex
it has processed. That is, after processing layer $D_{k}$, all elements in $F_k$ must be the beginning of a valid path, else they were removed. In particular, the algorithm remembers the valid path starting at $x$ as
$(x, \texttt{next}[x], \texttt{next}[\texttt{next}[x]], \ldots)$.
It is easy to verify that this sequence does indeed satisfy the conditions of \emph{valid paths} by inspecting lines \ref{lst:find-1} and \ref{lst:find-2}.

We also discuss what happens when the algorithm chooses to add a valid path to the partial augmenting set (i.e.\ in line~\ref{lst:valid-path-1} or \ref{lst:valid-path-2}).
If we are in Line~\ref{lst:valid-path-2}, we can directly apply \cref{lem:valid-path}.
Say we instead are in Line~\ref{lst:valid-path-1}, and some $a$ which was previously fresh has been added to $A_k$. The \texttt{RefineBA} call can only have increased
$A_k$ (that is $A_k \supseteq  A^{old}_k + a)$,
so $S-A_k+B_{k+1} + b\in \cI_1$ will holds for $b = \texttt{next}[a]$
and we can apply \cref{lem:valid-path} here too.

\paragraph{When no path is found.}
In the case when no valid path to add to the partial augmenting set is found, \texttt{RefinePath} must terminate with $|B_1| = |A_1| = \cdots = |B_{\ell+1}|$.
This is because the $\texttt{RefineAB}$
and $\texttt{RefineBA}$ will never select any new elements. That is $\texttt{RefineBA}$ will not change $A_k$ (as otherwise we enter the if-statement at line~\ref{lst:valid-path-1}),
and $\texttt{RefineAB}$ will not change $B_{k}$ (since if $b\in F_{2k-1}$
with $S-A_{k-1}+B_k+b\in \cI_1$ existed we would have entered the if-statement at line~\ref{lst:valid-path-2}).
We also remark that $\texttt{RefinePath}$ ends with $B_1$ being a maximal subset
of $D_1\setminus R_1$, as otherwise some $b$ would have been found in line~\ref{lst:check-2}.
Hence \cref{lem:invariants-maximality} implies that $(B_1, A_1, \ldots, B_{\ell+1})$ now forms a \emph{maximal} augmenting set.

\paragraph{Query complexity.}
The \texttt{RefineAB} and \texttt{RefineBA} calls will in total use $O(n)$ queries.
The independence checks at Lines~\ref{lst:check-1} and \ref{lst:check-2} happens at most once for each element, and thus use $O(n)$ queries in total.
Lines~\ref{lst:find-1} and \ref{lst:find-2} can be implemented using the binary-search-exchange-discovery \cref{lem:find-exchange}.
Hence Line~\ref{lst:find-1} will use, in total, $O(n\log r)$ queries
and Line~\ref{lst:find-2} will use, in total, $O(n\log r)$ queries (since each $a\in Q$ will be discovered at most once).
So we conclude that \cref{alg:refinepath} uses $O(n\log r)$ independence queries.

\subsection{Hybrid Algorithm}

Now we are finally ready to present the full algorithm of a phase, which is parameterized by a variable $p$.
The following algorithm is similar to that of \cite[Algorithm~12]{chakrabarty2019faster} but uses our improved \texttt{Refine} method  and finds individual paths using the \texttt{RefinePath} method.

\begin{algorithm}[ht]
\caption{Phase $\ell$} \label{alg:approx-phase}
\begin{algorithmic}[1]

\State
Calculate the distance layers by a BFS.

\State
Run $\texttt{Refine}$ (\cref{alg:refine}) until $|B_{1}|-|B_{\ell+1}|\leq p$, but at least once.

\State
Run $\texttt{RefinePath}$ (\cref{alg:refinepath}) until $(B_{1},A_1, \ldots B_{\ell+1})$ is maximal. Augment along it.%$\Phi_\ell$.

\end{algorithmic}
\end{algorithm}

\begin{lemma}
\label{lem:approx-phase-queries}
Except for line~1, \Cref{alg:approx-phase} uses $O(nr/p + np\log r)$ queries.\footnote{Compare this to $O(n^2/p + np\ell \log r)$ in \cite{chakrabarty2019faster}. The improvement from $n^2/p$ to $nr/p$
comes from the use of the new three-layer \texttt{RefineABA} method, and the (independent) improvement from $np\ell\log r$ to $np\log r$ comes from the use of the new \texttt{RefinePath} method.}
\end{lemma}
\begin{proof}
\Cref{lem:refine-types} tells us that \texttt{Refine} changes types of at least $p$
elements in even layers (i.e.\ elements in $S$) every time it is run, except maybe the last time.
Thus we only run \texttt{Refine} $O(|S|/p + 1)$ times.
Each call takes $O(n)$ queries (\cref{lem:refine-queries}),
for a total of $O(nr/p)$ queries in line 2 of the algorithm.

Now we argue that $B_1$ can never become larger than what it was just after line
2 was run. This is because $\texttt{Refine}$ will run at least once, and ends with a $\texttt{RefineABA}(0)$ call which in turn ends with a $\texttt{RefineAB(0)}$ call---which extends $B_1$ to be a maximal set in $D_1\setminus R_1$ for which $S+B_1\subseteq \cI_1$ holds.\footnote{Indeed, since $\cM_1$ is a matroid, all such maximal sets have the same size, so we can never obtain something larger later.}

\Cref{lem:refine-path} tells us that each (except the last) time $\texttt{RefinePath}$ is run, $B_{\ell+1}$ increases by~$1$. This can happen at most $p$ times,
so line 3 uses a total of $O(np\log r)$ queries.
\end{proof}

Now it is easy to prove \cref{thm:approx}, which we restate below.

\ThmApprox*

\begin{proof}%[Proof of \cref{thm:approx}]
Pick $p = \sqrt{r/\log r}$.\footnote{Compare this to $p = \sqrt{n\eps / \log r}$ in \cite{chakrabarty2019faster}.}
Then each phase will use $O(n\sqrt{r\log r})$ independence queries (by \cref{lem:approx-phase-queries}),
plus a total of $O(\frac{1}{\eps}n\log r)$ to run the BFS's across all
phases (see \cite{chakrabarty2019faster} for details on the BFS implementation). Since we need only run $O(\frac{1}{\eps})$ phases (by \cref{clm:dists} and
\cref{thm:max-inc-dist}), 
in total the algorithm will use $O(\frac{1}{\eps}n\sqrt{r\log r})$ queries.
\end{proof}

\section{Exact Matroid Intersection} \label{sec:exact}
In this section, we prove \cref{thm:exact} (restated below) by showing how our improved approximation algorithm leads to an improved exact algorithm when combined with the algorithms of \cite{quadratic2021}.

\ThmExact*

Approximation algorithms are great at finding the \emph{many, very short} augmenting
paths efficiently.
Blikstad-v.d.Brand-Mukhopadhyay-Nanongkai
\cite[Algorithm~2]{quadratic2021} very recently showed how to 
efficiently find the remaining \emph{few, very long} augmenting paths, with a randomized
algorithm using $\tO(n\sqrt{r})$ queries per augmentation
(or, with a slightly less efficient deterministic algorithm using $\tO(nr^{2/3})$ queries).
In the randomized $\tO(n^{6/5}r^{3/5})$-query exact algorithm of \cite[Algorithm~3]{quadratic2021}, the current bottleneck is the approximation algorithm used.
Replacing the use of the $\tO(n^{1.5}/\eps^{1.5})$-query approximation algorithm
from \cite{chakrabarty2019faster} with our improved version we obtain the more efficient randomized\footnote{The deterministic algorithm of \cref{thm:exact} is obtained in the same fashion but by using the deterministic version of the augmenting path finding algorithm \cite[Algorithm~2]{quadratic2021}.} $\tO(nr^{3/4})$-query \cref{alg:exact}.

\begin{algorithm}[ht]
\caption{Exact Matroid Intersection
\hfill(Modified version of \cite[Algorithm~3]{quadratic2021})}
\label{alg:exact}
\begin{algorithmic}[1]
\State
       Run the approximation algorithm (\cref{thm:approx})
       with $\eps = r^{-1/4}$
       to obtain a common independent set $S$ of size at least
       $(1-\eps)r = r - r^{3/4}$.
       \label{lst:line:approx}
\State Starting with $S$, run Cunningham's algorithm (as implemented by
       \cite{chakrabarty2019faster}), until the distance between $s$ and $t$
       becomes larger than $r^{3/4}$.
       \label{lst:line:cunningham}
\State Keep finding augmenting paths---one at a time---to augment along, using the randomized $O(n\sqrt{r}\log n)$-query algorithm of \cite[Algorithm~2]{quadratic2021}.
When no $(s,t)$-path can be found in the exchange graph, $S$ is a largest common independent set.
       \label{lst:line:augmenting-paths}
\end{algorithmic}
\end{algorithm}

\paragraph{Query complexity.}
We analyse the individual lines of \cref{alg:exact}.
\begin{description}
\item[Line \ref{lst:line:approx}.]
We see that the approximation algorithm uses $O(nr^{3/4}\log n)$ queries in line \ref{lst:line:approx}.

\item[Line \ref{lst:line:cunningham}.]
One need to
(i) compute distances up to $d = r^{3/4}$, and
(ii) perform at most $O(r^{3/4})$ augmentations.
\cite{chakrabarty2019faster, quadratic2021, nguyen2019note} show how to do (i) in
$O(nd\log n) = O(nr^{3/4}\log n)$ queries in total over all phases of Cunningham's algorithm, and how to do (ii)
using $O(n\log n)$ queries per augmentation (for a total of
$O(nr^{3/4}\log n)$ queries).

\item[Line \ref{lst:line:augmenting-paths}.]
By \cref{clm:dists}, line \ref{lst:line:augmenting-paths} runs
$O(r^{1/4})$ times---each using $O(n\sqrt{r}\log n)$ queries---for a total of $O(nr^{3/4}\log n)$ queries.
\end{description}

\begin{remark}
In \cref{alg:exact}, the bottleneck between line 1-2 and line 2-3 now matches
(which was not the case in \cite{quadratic2021}). This means that if one wants to improve
the algorithm by replacing the subroutines in line 1 and 3, one need to \textbf{both} improve the approximation algorithm (line 1) and the method to find a single augmenting-path (line 3).
\end{remark}

\section*{Acknowledgement}
This project has received funding from the European Research Council (ERC) under the European Unions Horizon 2020 research and innovation programme under grant agreement No 71567.

I also want to thank Danupon Nanongkai and Sagnik Mukhopadhyay for insightful discussions and their valuable comments throughout the development of this work.

\bibliography{biblio}

\newcommand{\etalchar}[1]{$^{#1}$}
\begin{thebibliography}{BvdBMN21}

\bibitem[AD71]{aignerD}
Martin Aigner and Thomas~A. Dowling.
\newblock Matching theory for combinatorial geometries.
\newblock {\em Transactions of the American Mathematical Society},
  158(1):231--245, 1971.

\bibitem[BvdBMN21]{quadratic2021}
Joakim Blikstad, Jan van~den Brand, Sagnik Mukhopadhyay, and Danupon Nanongkai.
\newblock Breaking the quadratic barrier for matroid intersection.
\newblock In {\em {STOC}}. {ACM}, 2021.

\bibitem[CLS{\etalchar{+}}19]{chakrabarty2019faster}
Deeparnab Chakrabarty, Yin~Tat Lee, Aaron Sidford, Sahil Singla, and
  Sam~Chiu{-}wai Wong.
\newblock Faster matroid intersection.
\newblock In {\em {FOCS}}, pages 1146--1168. {IEEE} Computer Society, 2019.

\bibitem[CQ16]{ChekuriQ16}
Chandra Chekuri and Kent Quanrud.
\newblock A fast approximation for maximum weight matroid intersection.
\newblock In {\em {SODA}}, pages 445--457. {SIAM}, 2016.

\bibitem[Cun86]{cunningham1986improved}
William~H. Cunningham.
\newblock Improved bounds for matroid partition and intersection algorithms.
\newblock {\em {SIAM} J. Comput.}, 15(4):948--957, 1986.

\bibitem[Edm70]{edmonds1970submodular}
Jack Edmonds.
\newblock Submodular functions, matroids, and certain polyhedra.
\newblock In {\em Combinatorial structures and their applications}, pages
  69--87. 1970.

\bibitem[Edm79]{edmonds1979matroid}
Jack Edmonds.
\newblock Matroid intersection.
\newblock In {\em Annals of discrete Mathematics}, volume~4, pages 39--49.
  Elsevier, 1979.

\bibitem[EDVJ68]{edmonds1968matroid}
Jack Edmonds, GB~Dantzig, AF~Veinott, and M~J{\"u}nger.
\newblock Matroid partition.
\newblock {\em 50 Years of Integer Programming 1958--2008}, page 199, 1968.

\bibitem[Law75]{Lawler75}
Eugene~L. Lawler.
\newblock Matroid intersection algorithms.
\newblock {\em Math. Program.}, 9(1):31--56, 1975.

\bibitem[LSW15]{lee2015faster}
Yin~Tat Lee, Aaron Sidford, and Sam~Chiu{-}wai Wong.
\newblock A faster cutting plane method and its implications for combinatorial
  and convex optimization.
\newblock In {\em {FOCS}}, pages 1049--1065. {IEEE} Computer Society, 2015.

\bibitem[Ngu19]{nguyen2019note}
Huy~L. Nguyen.
\newblock A note on cunningham's algorithm for matroid intersection.
\newblock {\em CoRR}, abs/1904.04129, 2019.

\end{thebibliography}

\end{document}